\documentclass[journal]{IEEETran}      

\IEEEoverridecommandlockouts
\usepackage{verbatim}
\usepackage{epsfig}
\usepackage{amsmath}
\usepackage{amsthm}
\usepackage{amssymb}
\usepackage{psfrag}
\usepackage[T1]{fontenc}
\usepackage[ruled]{algorithm2e}

\usepackage[usenames,dvipsnames]{pstricks}
\usepackage{pst-grad} 
\usepackage{pst-plot} 
\usepackage{etoolbox}
\usepackage[caption=false,font=footnotesize]{subfig}
\usepackage{url}
\usepackage[noadjust]{cite}
\usepackage{mathtools}
\usepackage{array}
\allowdisplaybreaks

\makeatletter
\patchcmd{\@makecaption}
  {\scshape}
  {}
  {}
  {}
\makeatother
\newtheorem{theorem}{Theorem}

\usepackage{dsfont}

\title{
Transactive Framework for Dynamic Energy Storage Allocation for Critical Load Management
}

\author{Arnab Dey$^{1}$, Vivek Khatana$^{1}$, Ankur Mani$^{2}$  and Murti V. Salapaka$^{1}$
\thanks{This work is supported by Advanced Research Projects Agency-Energy OPEN through the project titled "Rapidly Viable Sustained Grid" via grant no. DE-AR0001016.}
\thanks{$^{1}$ Arnab Dey \{{\tt\small dey00011@umn.edu}\}, Vivek Khatana \{{\tt\small khata010@umn.edu}\}, Murti V. Salapaka\{{\tt\small murtis@umn.edu}\} are with Department of Electrical and Computer Engineering, University of Minnesota, Twin Cities, USA, and $^{2}$ Ankur Mani \{{\tt\small amani@umn.edu\}} is with Department of Department of Industrial and Systems Engineering, University of Minnesota, Twin Cities, USA,
}
}

\begin{document}

\maketitle
\thispagestyle{empty}
\pagestyle{empty}

\begin{abstract}
Increased penetration of Distributed Energy Resources (DER) and Renewable Energy Systems (RES) transforming the conventional distribution grid into a transactive framework supervised by a distribution system operator (DSO). Although the emerging transactive energy management techniques improve the grid reliability, the inherent uncertainty of RES poses a challenge in meeting the power demand of the critical infrastructure in the microgrid unless sufficient battery energy storage is maintained. However, maintaining expensive battery storage increases the operating cost of the DSO. In this article, we propose a cost-effective dynamic resource allocation strategy to optimize the battery reserve requirement while ensuring the critical demand is met with a provable guarantee. Our proposed scheme enables the DSO to optimize the RES and battery reserve allocation to eliminate the risk of over or underproduction. We present numerical simulations under three different scenarios of multiple microgrids with uncertain renewable generation. The simulation results demonstrate the efficacy of the proposed transactive stochastic control algorithm.
\end{abstract}
\section{INTRODUCTION}
\IEEEPARstart{I}n recent times, the proliferation of distributed energy resources (DER), propelled by a mounting interest towards the generation and utilization of clean renewable energies, has fostered a paradigm change in the operation and control of power distribution networks. Renewable energy sources (RES), like solar photo-voltaic systems and wind energy systems, reduce the adverse environmental effects of traditional energy sources; integration of DERs also offer improved grid assistance, efficiency and reliability. As renewables are expected to contribute $60\%$ of total generation by $2050$ \cite{irena2019energytransformation}, a rapid growth of `prosumers', equipped with distributed RES, is envisaged to transform the conventional energy systems (CES) to \textit{transactive energy systems} (TES). According to U.S. Department of Energy GridWise Architecture Council, TES are defined as the systems of coordinated control mechanisms that allow dynamic balance of supply and demand, using value as a key operational parameter \cite{gridwise}. However, due to the uncertain and intermittent nature of the RES, maintaining the balance between power supply and demand in a distributed framework is challenging if extra measures such as mechanisms for ancillary services are not present \cite{irena, yang2018battery}. Moreover, climate related catastrophic events are increasing in frequency and magnitude \cite{kishore2018mortality,campbell2012weather}. In this context, collaborative operation across multiple critical facilities allows for pooling renewable resources and reduces uncertainty and impact of idiosyncratic generation shocks due to catastrophic failure at a facility. For such critical infrastructures, it is important to guarantee needed power and thus managing uncertainty of renewable energy sources needs to be addressed. Therefore, collaborative operation, intertwined with uncertainty mitigation techniques, provides an attractive solution when RES are used to supply power.

Energy storage allocation for demand-supply balance, considering fluctuating renewable generation, is of significant interest presently to the researchers. In \cite{salinas2013dynamic}, a dynamic energy management scheme, considering stochastic load demands and renewable generations of multiple prosumers, is proposed based on Lyapunov optimization theory. An alternating direction method of multipliers (ADMM) based optimization method is proposed in \cite{liu2019secure} with the focus of minimizing energy trading cost in a TES framework. Prior works have considered game theoretic approaches \cite{divshali2017multi, zhang2019dynamic}, stochastic programming \cite{wang2019incentive} and, mixed-integer linear programming \cite{renani2017optimal} to maximize the profit of each agent in TES with high penetration of RES; however without explicitly modelling the RES output and considering real-time power balance. In \cite{hu2019agent}, a multi-agent optimization technique is presented to improve demand-supply balance in distribution networks with DERs. An analytical approach, based on cost-benefit analysis, is devised in \cite{akter2016hierarchical} to improve energy sharing instead of profit maximization, within a residential microgrid, however, stochasticity of the renewable generation was not addressed. In \cite{xu2019novel}, an intelligent network of agents, connected to a DC microgrid, is considered to conduct power trading based on min-consensus algorithm. Here, each agent broadcasts power demand request and based on the request received from the neighbors, optimal power dispatch solution is obtained. Reference \cite{sajjadi2016transactive} provides a case study of a TES where the distribution system operator controls the energy trading between transactive nodes in a day-ahead market and benefits of the participants are estimated based on locational marginal price. In \cite{zhang2016optimal,maleki2016optimal}, authors have proposed Monte Carlo simulation methodology for optimal allocation of battery energy storage under fluctuating wind and solar power. Optimization algorithms emphasizing on overall investment and operational cost reduction is the focus of \cite{zhang2016optimal,maleki2016optimal,geem2012size,ghaffari2015energy,hedman2006comparing}. To address the stochastic nature of the solar radiation and wind, \cite{verdejo2016stochastic,olsson2010modeling, dong2016application,salameh1995photovoltaic} incorporate probabilistic approaches to model the renewable power output.  However, most of the TES framework found in literature do not consider explicit probabilistic models for the renewable generation. Moreover, most of the existing literature either focus on cost optimization for energy trading in TES without addressing the energy storage requirement and supply guarantee required by the critical infrastructures (e.g. hospitals), or, solely consider energy storage optimization without exploring the benefit of utilization of the same in transactive architecture. Here, an approach for optimal energy storage allocation to mitigate the uncertainty of meeting load demands of critical infrastructures in a TES, due to stochastic nature of renewable generations, is required.

To this end, in this article we focus on a TES framework consisting of centralized battery storage and multiple number of small-scale microgrids with critical infrastructures that require guarantee of supplying the demand requested at a future time. Each microgrid is equipped with RES which are uncertain in nature, and has no recourse to local energy storage, thus posing a risk of not being able to meet the local critical demand. To mitigate this uncertainty, the distribution system operator (DSO), which supervises the whole TES, allocates battery storage to support the microgrids and also enables the microgrids to trade renewable energies with each other. The focus of the article is on the optimal allocation of battery storage by the DSO to meet the power demand of each microgrid at the requested time.

\noindent The major contributions are summarized below:

\noindent(i) We study a TES framework with multiple prosumers managed by DSO to provide a solution to the problem of meeting power demand of each prosumer in real-time with provable guarantees, in the presence of stochasticity of the renewable energy sources. Such a guarantee is essential for critical applications. 

\noindent(ii) We also provide exact solutions for continuous time resource allocations for the DSO to manage renewable and battery resources, eliminating risks of underproduction and overproduction in the presence of uncertainty in renewable generation. Our solution builds upon stochastic control techniques and we also provide an approximation algorithm for discrete time resource management. We remark that such a solution is pertinent for supporting critical infrastructures.

The rest of the paper is organized as follows: description of the TES architecture is given in Section~\ref{sec:syst_desc}. A formal problem formulation is introduced in Section~\ref{sec:prob_form}. Section~\ref{sec:solu_meth} describes the proposed resource allocation algorithms and the simulation results are shown in Section~\ref{sec:simu_resu}. Section~\ref{sec:conc_lude} concludes the article with concluding remarks.
\section{System Description}\label{sec:syst_desc}
Renewable based large distribution networks, with centralized battery storage \cite{shadmand2014multi}, consist of several small-scale microgrids geographically apart from each other, under the supervision of DSOs. We consider such a distributed framework constituting a TES, where the DSO can operate the network more efficiently than conventional one by utilising the ability of cooperation between the microgrids. Fig.~\ref{fig:tran_arch} shows the hierarchical architecture of the TES under consideration, which contains $N_m$ microgrids and a DSO which controls the energy transaction among the microgrids and dispatches the centralized battery units, based on the requirements, to maintain system reliability. Each microgrid has RES to supply its local demands. However, the uncertain nature of the RES poses the risk of either power deficit, when generation is less than the demand, or excess power that forces curtailment, when generation exceeds the demand. In CES, this risk is minimized by using independent energy storage systems such as batteries for individual microgrids \cite{faccio2018state}, thus requiring large ancillary battery energy storage systems (BESS) \cite{birnie2014optimal,olatomiwa2016energy, oh2020theoretical}. DSO controls the energy transactions and the dispatch of batteries when necessary. A communication layer is used to transmit the real-time power measurements of the net generation and demands of the individual microgrids to aid the DSO for efficient resource allocation. To reduce the energy cost, the DSO needs to estimate the optimal battery storage while ensuring that the load demands of the microgrids are met, irrespective of the irregularities of the renewable sources. In the subsequent sections we formally introduce the problem and provide a solution.
\begin{figure*}[h!]
    \centering
    \includegraphics[scale=0.463,trim={0.0cm 0.0cm 0.0cm 0.0cm},clip]{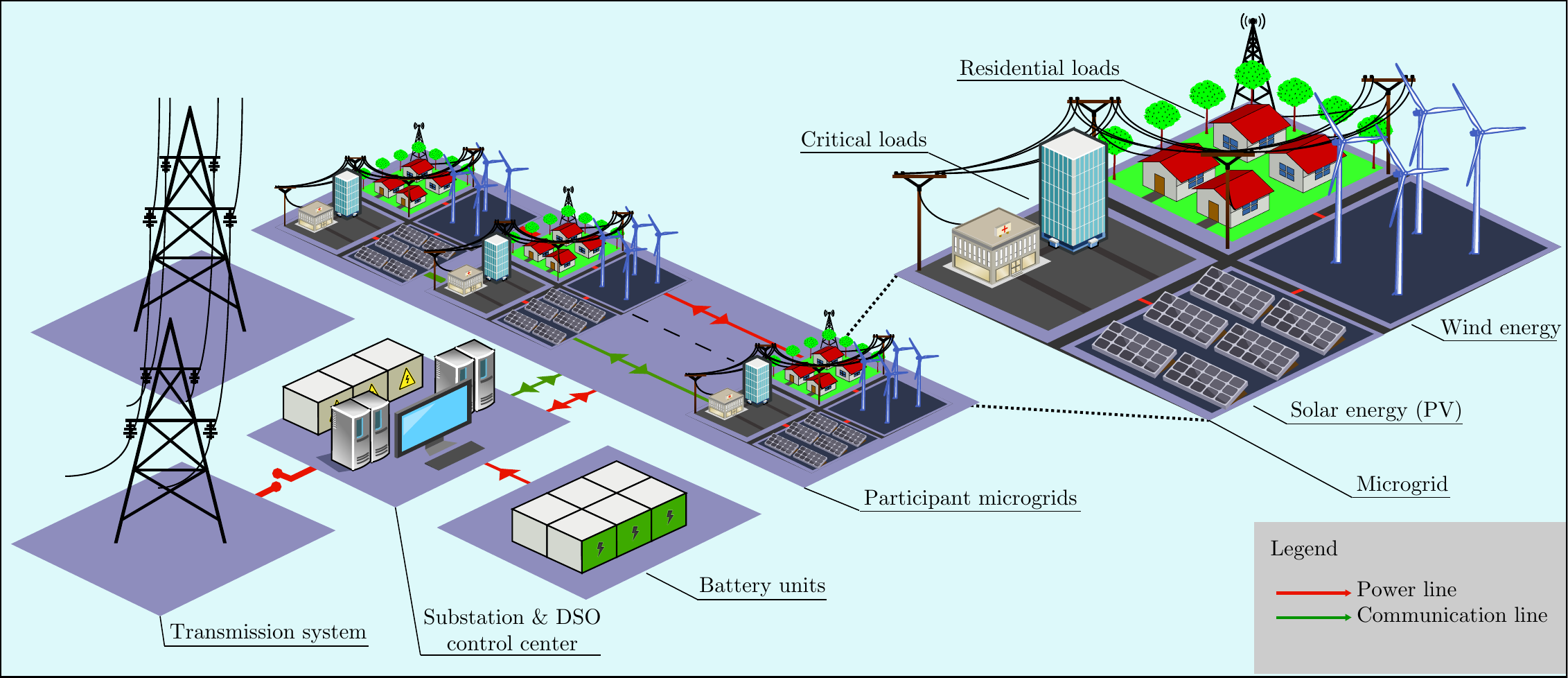}
    \caption{TES architecture with centralized battery units and distributed microgrids with renewable generations and critical loads}
    \label{fig:tran_arch}
\end{figure*}
\section{Problem formulation}\label{sec:prob_form}
Let $\mathcal{I} = \{1,2,\ldots,N_m\}$ denote the set of  microgrids present in the distribution network. Each microgrid has its own renewable generation unit (ReGU) producing $P_{gi}(t), i \in \mathcal{I}$, units of power at time $t$, which is stochastic in nature and is modeled by Geometric Brownian Motion (GBM) given by:
\begin{align}\label{eq:gbm}
    \text{d}P_{gi}(t) = \mu_{gi}P_{gi}(t)\text{d}t+\sigma_{gi}P_{gi}(t)\text{d}W_{ti},
\end{align}
with an initial condition of $P_{gi}(0) > 0$, for all $i \in \mathcal{I}$. The constants $\mu_{gi}$ and $\sigma_{gi}$ are the drift and the volatility terms respectively and $W_{ti}$ denotes a Weiner process. We empirically verify this model assumption with renewable generation data in Section~\ref{sec:simu_resu}. The Weiner processes $W_{ti}$, $W_{tj}$ are correlated according to
\begin{align}\label{eq:corr_wtij}
    \text{d}W_{ti}\text{d}W_{tj} = \rho_{ij}\text{d}t,\ i,j \in \mathcal{I},
\end{align}
where $\rho_{ij}$ is the correlation coefficient. Therefore,
\begin{align}\label{eq:corr_pgij}
    \text{d}P_{gi}\text{d}P_{gj} = \sigma_{i}\sigma_{j}\rho_{ij}P_{gi}P_{gj}\text{d}t,\ i,j \in \mathcal{I}.
\end{align}
Each microgrid also postulates a constant critical power demand of $D_{ci}$ units, $i \in \mathcal{I}$, at a future time denoted as $T_f$. In a transactive setting, the objective of the system is to satisfy power demand and supply balance at a predefined future time, $T_f$, \textit{almost surely} with minimum amount of resources collectively. In case of $P_{gi}(T_f) \geq D_{ci}$ for $i^{th}$ microgrid, it can ensure sustained operation of its own load demand. On the other hand when $P_{gi}(T_f) < D_{ci}$, to avoid the risk of microgrid $i$ not being able to meet the power demand of its critical loads, it can be supported by the energy storage systems maintained by the DSO. Here, the microgrids participate in a contract with the DSO, which controls the energy transactions and the battery allocation for the whole system by {\it independently} maintaining a portfolio consisting of fractions of ReGUs from each participant microgrid and centralized battery units. In particular, as shown in Fig.~\ref{fig:tran_arch}, the ReGUs in the portfilo of the DSO are physically located at the premises of the individual microgrids, thus following same dynamics as given in~(\ref{eq:gbm}), while the battery storage, which is shared by all the microgrids, is fully controlled and allocated by the DSO.

In the case of CES where power sharing among the microgrids is not possoble and each microgrid is treated independently by the DSO \textit{i.e.} available power in the portfolio maintained by the DSO for each microgrid is:
\begin{align}\label{eq:conv_port}
    \hat{V}_i(P_{gi}(t),t) = \hat{a}_i(t)P_{gi}(t)+\hat{b}_i(t)P_b,\ i \in \mathcal{I},
\end{align}
and $ 0 \leq t \leq T_f$, where $\hat{a}_i(t)$ and $\hat{b}_i(t)$ are the number of ReGUs and battery units allocated for microgrid $i$. In case of CES, the constraint of \textit{rated power conservation} is given by:
\begin{align}\label{eq:ces_pconv}
    \text{d}\hat{a}_i(t)P_{gi}(t) + \text{d}\hat{b}_i(t) P_b = 0, \ \mbox{for all} \  i \in \mathcal{I},
\end{align}
and the total amount of battery units maintained by DSO is:
\begin{align*}
    \hat{b}(t) = \sum_{i=1}^{N_m}\hat{b}_i(t),\ \forall t \in [0,T_f]. 
\end{align*}
To ensure the power demand of $i^{th}$ microgrid is met at the requested time $T_f$, the DSO needs to supply the deficit of $D_{ci} - P_{gi}(T_f)$ units of power from the portfolio if the microgrid is not capable of supplying its own demand from the renewable generation, and $0$ otherwise. Therefore the power portfolio has to satisfy the following terminal condition,
\begin{align}\label{eq:conv_tcon}
   \hspace{-0.1in} \hat{V}_i(P_{gi}(T_f),T_f) = 
    \begin{cases}
    0                        & \text{if $P_{gi}(T_f) \geq D_{ci}$}\\
    D_{ci} - P_{gi}(T_f)           & \text{if $P_{gi}(T_f) < D_{ci}$}.
    \end{cases}
\end{align}
The DSO needs to decide the number $\hat{a}_i$ and $\hat{b}_i$ such that $\hat{V}_i$ satisfy the terminal condition~\eqref{eq:conv_tcon} under the constraint~\eqref{eq:ces_pconv}.
\section{Solution Methodology}\label{sec:solu_meth} 
\subsection{Conventional Energy System}\label{subsec:solu_conv}

As the amount of renewable generation at time $T_f$ is not deducible and also not known a priori due to the inherent uncertainty \cite{black1973pricing}, we provide a policy of maintaining the number of battery units, $\hat{b}_i(t)$ and the number of ReGUs, $\hat{a}_i(t)$ for $t \in [0, T_f)$, satisfying the terminal condition~(\ref{eq:conv_tcon}), such that the critical demand of $D_{ci},\ i \in \mathcal{I}$, units is met \textit{almost-surely} at time $T_f$ in the following theorem.
\begin{theorem}\label{thm:bs_sol}
Under the rated power conservation constraint~(\ref{eq:ces_pconv}), suppose $\hat{a}_i(t)$ and $\hat{b}_i(t)$ are given by:
\begin{align*}
    \hat{a}_i(t) &= -\Phi\bigg[\textstyle\frac{\ln\left(\frac{D_{ci}}{P_{gi}(t)}\right)-\textstyle\frac{\sigma_{gi}^2}{2}(T_f - t)}{\sigma_{gi} \sqrt{(T_f - t)}}\bigg],\\ 
    \hat{b}_i(t) &= \textstyle\frac{D_{ci}}{P_b}\Phi\bigg[ \frac{\ln\left(\frac{D_{ci}}{P_{gi}(t)}\right)+\textstyle\frac{\sigma_{gi}^2}{2}(T_f - t)}{\sigma_{gi} \sqrt{(T_f - t)}}\bigg],\ t \in [0,T_f)
\end{align*}
where, $\Phi(\cdot)$ is the standard normal cumulative distribution function, then the terminal condition specified by~(\ref{eq:conv_tcon}) is satisfied almost surely at time $T_f$.
\end{theorem}
\begin{proof}
Please see supplementary material.
\end{proof}
It follows that, in case of conventional grid, the total amount of battery units required by the DSO for all $t \in [0, T_f)$ is:
\begin{align}
    \hat{b}(t) =\sum_{i=1}^{N_m} \textstyle\frac{D_{ci}}{P_b}\Phi\bigg[ \frac{\ln\left(\frac{D_{ci}}{P_{gi}(t)}\right)+\textstyle\frac{\sigma_{gi}^2}{2}(T_f - t)}{\sigma_{gi} \sqrt{(T_f - t)}}\bigg],
\end{align}
while the total power portfolio value at $t=T_f$ is:
\begin{align*}
    \hat{V}(P_{g1}(T_f),\ldots,P_{gN_m}, T_f) = \sum_{i=1}^{N_m}\text{max}(D_{ci}-P_{gi}(T_f),0).
\end{align*}
\section{Simulation results}\label{sec:simu_resu}
\subsection{Case study $1$: $P_{g1}(T_f) \geq D_{c1}, P_{g2}(T_f) \geq D_{c1}$}\label{subsec:simu_upup}
Fig.~\ref{fig:pg_upup} illustrates the scenario where the renewable generations of both the microgrids are sufficient to meet the respective load demands at $t=T_f$. We remark that the algorithm and the methodology has no knowledge of the terminal state {\it  a priori} at any time $t<T_f$. It is seen that the battery units requirement as well as the power portfolio value become $0$ at $t=T_f$ as expected. Similarly, the other cases can also be shown which we omit here to conserve space.
\begin{figure}[!h]
     \centering
     \subfloat[Battery units requirement: No battery is required at $T_f$ as both the microgrids have sufficient generation to supply own demand. Standard deviation of $10,000$ realizations of $P_{g1}(t)$ and $P_{g2}(t)$ are shown in the shaded regions with the mean values plotted with solid lines. Zoomed inset shows the $95$\% confidence interval (CI) around the mean of battery requirement. The narrow CI indicates high confidence in average battery requirement of our proposed scheme.]{\includegraphics[scale=1.0,trim={0.4cm 0cm 0cm 0cm},clip]{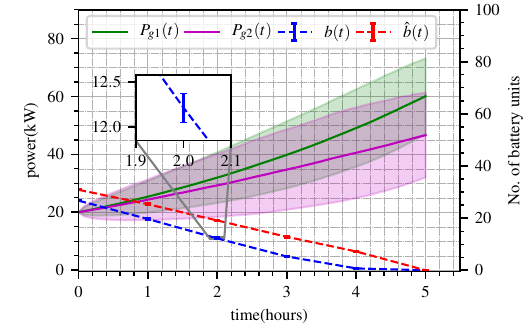}\label{fig:b_pg_upup}}\hspace{0.5cm}
     \subfloat[Power portfolio plot: No amount of power to be maintained in the portfolios at time $T_f$ as both the microgrids are capable of supplying their own demand. The narrow $95$\% CI around the mean of power portfolio value, as shown in zoomed inset, indicates high confidence in the average power requirement in the portfolio of our proposed scheme.]{\includegraphics[scale=1.0,trim={0.4cm 0cm 0cm 0.07cm},clip]{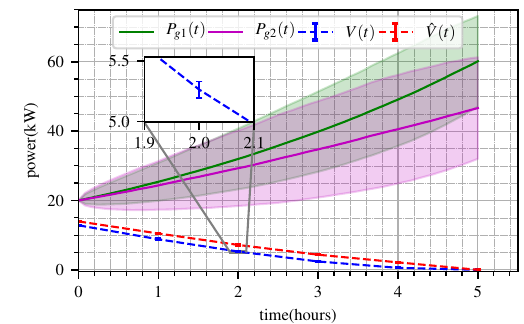}\label{fig:v_pg_upup}}
     \caption{Simulation result for $10000$ realizations of $P_{g1}(t), P_{g2}(t)$ based on~(\ref{eq:gbm}) with $D_{c1}=20\text{ kW},D_{c2}=25$ kW, $T_f=5$ hrs and $P_{g1}(T_f)>D_{c1}, P_{g2}(T_f)>D_{c2}$ which is unknown \textit{a priori} at time $t < T_f$.}
     \label{fig:pg_upup}
\end{figure}

\noindent The results corroborate the effectiveness and the utility of the framework presented with respect to battery resources requirements while providing provable guarantees that the power needs will be met at the end of the prescribed horizons. Therefore, it can be concluded that our proposed algorithm enables DSO to mitigate the uncertainty of supplying the demand of the microgrids at the requested time while reducing the battery requirement, thus enhancing grid reliability and efficiency.
\section{Conclusion}\label{sec:conc_lude}
This article develops a novel TES framework containing multiple microgrids, equipped with DER and critical infrastructure loads, and an DSO with centralized battery storage, with the focus of efficient energy storage allocation in real-time, to meet the load demands of the microgrids at the requested time with \textit{almost sure guarantee}. Explicit stochastic model of wind generation is considered to represent the uncertainties of renewable generations. The proposed algorithms mitigate the uncertainties of meeting the load demands of the microgrids, due to intermittent renewable generations, at the requested time, with \textit{almost sure} guarantee.
\appendices
\section{Probabilities and change factors calculation}\label{app:unkn_prob}
Applying Ito's lemma \cite{gardiner2009stochastic} to $\text{d}(\log P_{gi}(t))$ and ignoring \textit{h.o.t.},
\begin{align*}
    \textstyle \text{d}(\log P_{gi}(t)) &= \textstyle \frac{\partial (\log P_{gi}(t))}{\partial P_{gi}(t)}\text{d}P_{gi}(t) + \frac{\partial^2 (\log P_{gi}(t))}{\partial P_{gi}^2(t)}\frac{(\text{d}P_{gi}(t))^2}{2}
\end{align*}
For all $i \in \mathcal{I}$,
\begin{align}
    \textstyle \text{d}(\log P_{gi}(t)) &= \textstyle \frac{\text{d}P_{gi}(t)}{P_{gi}(t)}-\frac{1}{2}\frac{1}{P_{gi}^2(t)}\sigma_{gi}^2P_{gi}^2(t)\text{d}t \nonumber \\
    &= \textstyle - \frac{\sigma_{gi}^2}{2}\text{d}t + \sigma_{gi} \text{d}\hat{W}_{ti} \nonumber\\
    &= \textstyle - \frac{\sigma_{gi}^2}{2}\text{d}t + \sigma_{gi} \sqrt{dt}Z,
\end{align}
where $Z \sim \mathcal{N}(0,1)$. Let $X_n^i = \log P_{gi}(t_n)$. Therefore,
\begin{align*}
    \textstyle \hat{\mathbb{E}}[X_{n+1}^i-X_n^i] = -\frac{\sigma_{gi}^2}{2}\Delta t,\ 
    \text{Var}(X_{n+1}^i-X_n^i) = \sigma_{gi}^2\Delta t.
\end{align*}
Under the transformed probability measure~(\ref{eq:neut_meas}), from~(\ref{eq:corr_wtij}),
\begin{align*}
    \text{Corr}((X_{n+1}^i-X_n^i),(X_{n+1}^j-X_n^j)) = \rho_{ij}\sigma_{gi}\sigma_{gj}\Delta t
\end{align*}
Let $\mathcal{I}_i$ denote the set of indices of probabilities associated with upward movement of $P_{gi}(t_n)$ and $\mathcal{I}_{ij} = (\mathcal{I}_i \cap \mathcal{I}_j) \cup (\mathcal{I}_i^c \cap \mathcal{I}_j^c)$, for all $i,j \in \mathcal{I}$. Denoting $h_i = \log u_i, i \in \mathcal{I}$, 
\begin{subequations}
\begin{align}
    \textstyle -\frac{\sigma_{gi}^2}{2}\Delta t &= \textstyle h_i(\sum\limits_{k \in \mathcal{I}_i}P_k-\sum\limits_{k \notin \mathcal{I}_i}P_k)\\
    \textstyle \sigma_{gi}^2\Delta t &= \textstyle h_i^2(\sum\limits_{k=1}^{2^{N_m}}P_k) -h_i^2(\sum\limits_{k \in \mathcal{I}_i}P_k-\sum\limits_{k \notin \mathcal{I}_i}P_k)^2\\
    \textstyle \rho_{ij}\sigma_{gi}\sigma_{gj}\Delta t &= \textstyle h_ih_j(\sum\limits_{k \in \mathcal{I}_{ij}}P_k-\sum\limits_{k \notin \mathcal{I}_{ij}}P_k)\\
    1 &= \textstyle \sum\limits_{k=1}^{2^{N_m}}P_k
\end{align}
\end{subequations}
We can solve these set of equations for all $i,j \in \mathcal{I}$ to find the values of $P_k, k \in \{1,2,\ldots,2^{N_m}\}$, $u_i,d_i, i \in \mathcal{I}$.

\bibliography{references}

\begin{thebibliography}{10}
\providecommand{\url}[1]{#1}
\csname url@samestyle\endcsname
\providecommand{\newblock}{\relax}
\providecommand{\bibinfo}[2]{#2}
\providecommand{\BIBentrySTDinterwordspacing}{\spaceskip=0pt\relax}
\providecommand{\BIBentryALTinterwordstretchfactor}{4}
\providecommand{\BIBentryALTinterwordspacing}{\spaceskip=\fontdimen2\font plus
\BIBentryALTinterwordstretchfactor\fontdimen3\font minus
  \fontdimen4\font\relax}
\providecommand{\BIBforeignlanguage}[2]{{%
\expandafter\ifx\csname l@#1\endcsname\relax
\typeout{** WARNING: IEEEtran.bst: No hyphenation pattern has been}%
\typeout{** loaded for the language `#1'. Using the pattern for}%
\typeout{** the default language instead.}%
\else
\language=\csname l@#1\endcsname
\fi
#2}}
\providecommand{\BIBdecl}{\relax}
\BIBdecl

\bibitem{irena2019energytransformation}
\BIBentryALTinterwordspacing
IRENA, ``Global energy transformation: A roadmap to 2050,'' 2019, [Accessed 31
  August 2020]. [Online]. Available:
  \url{https://www.irena.org/-/media/Files/IRENA/Agency/Publication/2019/Apr/IRENA_Global_Energy_Transformation_2019.pdf}
\BIBentrySTDinterwordspacing

\bibitem{gridwise}
\BIBentryALTinterwordspacing
GridWise, ``Transactive energy systems research, development and deployment
  roadmap,'' 2018, [Accessed 31 August 2020]. [Online]. Available:
  \url{https://www.gridwiseac.org/pdfs/pnnl_26778_te_roadmap_dec_2018.pdf}
\BIBentrySTDinterwordspacing

\bibitem{irena}
\BIBentryALTinterwordspacing
IRENA, ``Battery storage for renewables: market status and technology
  outlook,'' 2015, [Accessed 24 March 2020]. [Online]. Available:
  \url{https://www.irena.org/documentdownloads/publications/irena_battery_storage_report_2015.pdf}
\BIBentrySTDinterwordspacing

\bibitem{yang2018battery}
Y.~Yang, S.~Bremner, C.~Menictas, and M.~Kay, ``Battery energy storage system
  size determination in renewable energy systems: A review,'' \emph{Renewable
  and Sustainable Energy Reviews}, vol.~91, pp. 109--125, 2018.

\bibitem{kishore2018mortality}
N.~Kishore, D.~Marqu{\'e}s, A.~Mahmud, M.~V. Kiang, I.~Rodriguez, A.~Fuller,
  P.~Ebner, C.~Sorensen, F.~Racy, J.~Lemery \emph{et~al.}, ``Mortality in
  puerto rico after hurricane maria,'' \emph{New England journal of medicine},
  vol. 379, no.~2, pp. 162--170, 2018.

\bibitem{campbell2012weather}
R.~J. Campbell and S.~Lowry, ``Weather-related power outages and electric
  system resiliency.''\hskip 1em plus 0.5em minus 0.4em\relax Congressional
  Research Service, Library of Congress Washington, DC, 2012.

\bibitem{salinas2013dynamic}
S.~Salinas, M.~Li, P.~Li, and Y.~Fu, ``Dynamic energy management for the smart
  grid with distributed energy resources,'' \emph{IEEE Transactions on Smart
  Grid}, vol.~4, no.~4, pp. 2139--2151, 2013.

\bibitem{liu2019secure}
Y.~Liu, H.~B. Gooi, Y.~Li, H.~Xin, and J.~Ye, ``A secure distributed
  transactive energy management scheme for multiple interconnected microgrids
  considering misbehaviors,'' \emph{IEEE Transactions on Smart Grid}, vol.~10,
  no.~6, pp. 5975--5986, 2019.

\bibitem{divshali2017multi}
P.~H. Divshali, B.~J. Choi, and H.~Liang, ``Multi-agent transactive energy
  management system considering high levels of renewable energy source and
  electric vehicles,'' \emph{IET Generation, Transmission \& Distribution},
  vol.~11, no.~15, pp. 3713--3721, 2017.

\bibitem{zhang2019dynamic}
J.~Zhang, B.-C. Seet, and T.~T. Lie, ``Dynamic energy scheduling for virtual
  power plant with prosumer resources using game theory,'' in \emph{2019 IEEE
  Power \& Energy Society General Meeting (PESGM)}.\hskip 1em plus 0.5em minus
  0.4em\relax IEEE, 2019, pp. 1--5.

\bibitem{wang2019incentive}
J.~Wang, H.~Zhong, J.~Qin, W.~Tang, R.~Rajagopal, Q.~Xia, and C.~Kang,
  ``Incentive mechanism for sharing distributed energy resources,''
  \emph{Journal of Modern Power Systems and Clean Energy}, vol.~7, no.~4, pp.
  837--850, 2019.

\bibitem{renani2017optimal}
Y.~K. Renani, M.~Ehsan, and M.~Shahidehpour, ``Optimal transactive market
  operations with distribution system operators,'' \emph{IEEE Transactions on
  Smart Grid}, vol.~9, no.~6, pp. 6692--6701, 2017.

\bibitem{hu2019agent}
S.~Hu, Y.~Xiang, J.~Liu, C.~Gu, X.~Zhang, Y.~Tian, Z.~Liu, and J.~Xiong,
  ``Agent-based coordinated operation strategy for active distribution network
  with distributed energy resources,'' \emph{IEEE Transactions on Industry
  Applications}, vol.~55, no.~4, pp. 3310--3320, 2019.

\bibitem{akter2016hierarchical}
M.~N. Akter, M.~A. Mahmud, and A.~M. Oo, ``A hierarchical transactive energy
  management system for microgrids,'' in \emph{2016 IEEE Power and Energy
  Society General Meeting (PESGM)}.\hskip 1em plus 0.5em minus 0.4em\relax
  IEEE, 2016, pp. 1--5.

\bibitem{xu2019novel}
Y.~Xu, H.~Sun, and W.~Gu, ``A novel discounted min-consensus algorithm for
  optimal electrical power trading in grid-connected dc microgrids,''
  \emph{IEEE Transactions on Industrial Electronics}, vol.~66, no.~11, pp.
  8474--8484, 2019.

\bibitem{sajjadi2016transactive}
S.~M. Sajjadi, P.~Mandal, T.-L.~B. Tseng, and M.~Velez-Reyes, ``Transactive
  energy market in distribution systems: A case study of energy trading between
  transactive nodes,'' in \emph{2016 North American Power Symposium
  (NAPS)}.\hskip 1em plus 0.5em minus 0.4em\relax IEEE, 2016, pp. 1--6.

\bibitem{zhang2016optimal}
Y.~Zhang, Z.~Y. Dong, F.~Luo, Y.~Zheng, K.~Meng, and K.~P. Wong, ``Optimal
  allocation of battery energy storage systems in distribution networks with
  high wind power penetration,'' \emph{IET Renewable Power Generation},
  vol.~10, no.~8, pp. 1105--1113, 2016.

\bibitem{maleki2016optimal}
A.~Maleki, M.~G. Khajeh, and M.~Ameri, ``Optimal sizing of a grid independent
  hybrid renewable energy system incorporating resource uncertainty, and load
  uncertainty,'' \emph{International Journal of Electrical Power \& Energy
  Systems}, vol.~83, pp. 514--524, 2016.

\bibitem{geem2012size}
Z.~W. Geem, ``Size optimization for a hybrid photovoltaic--wind energy
  system,'' \emph{International Journal of Electrical Power \& Energy Systems},
  vol.~42, no.~1, pp. 448--451, 2012.

\bibitem{ghaffari2015energy}
R.~Ghaffari and B.~Venkatesh, ``Energy reserve trade optimization for wind
  generators using black and scholes options in small-size power systems,''
  \emph{Canadian Journal of Electrical and Computer Engineering}, vol.~38,
  no.~2, pp. 66--76, 2015.

\bibitem{hedman2006comparing}
K.~W. Hedman and G.~B. Shebl{\'e}, ``Comparing hedging methods for wind power:
  Using pumped storage hydro units vs. options purchasing,'' in \emph{2006
  International Conference on Probabilistic Methods Applied to Power
  Systems}.\hskip 1em plus 0.5em minus 0.4em\relax IEEE, 2006, pp. 1--6.

\bibitem{verdejo2016stochastic}
H.~Verdejo, A.~Awerkin, E.~Saavedra, W.~Kliemann, and L.~Vargas, ``Stochastic
  modeling to represent wind power generation and demand in electric power
  system based on real data,'' \emph{Applied Energy}, vol. 173, pp. 283--295,
  2016.

\bibitem{olsson2010modeling}
M.~Olsson, M.~Perninge, and L.~S{\"o}der, ``Modeling real-time balancing power
  demands in wind power systems using stochastic differential equations,''
  \emph{Electric Power Systems Research}, vol.~80, no.~8, pp. 966--974, 2010.

\bibitem{dong2016application}
J.~Dong, A.~A. Malikopoulos, S.~M. Djouadi, and T.~Kuruganti, ``Application of
  optimal production control theory for home energy management in a micro
  grid,'' in \emph{2016 American Control Conference (ACC)}.\hskip 1em plus
  0.5em minus 0.4em\relax IEEE, 2016, pp. 5014--5019.

\bibitem{salameh1995photovoltaic}
Z.~M. Salameh, B.~S. Borowy, and A.~R. Amin, ``Photovoltaic module-site
  matching based on the capacity factors,'' \emph{IEEE transactions on Energy
  conversion}, vol.~10, no.~2, pp. 326--332, 1995.

\bibitem{shadmand2014multi}
M.~B. Shadmand and R.~S. Balog, ``Multi-objective optimization and design of
  photovoltaic-wind hybrid system for community smart dc microgrid,''
  \emph{IEEE Transactions on Smart Grid}, vol.~5, no.~5, pp. 2635--2643, 2014.

\bibitem{faccio2018state}
M.~Faccio, M.~Gamberi, M.~Bortolini, and M.~Nedaei, ``State-of-art review of
  the optimization methods to design the configuration of hybrid renewable
  energy systems (hress),'' \emph{Frontiers in Energy}, vol.~12, no.~4, pp.
  591--622, 2018.

\bibitem{birnie2014optimal}
D.~P. Birnie~III, ``Optimal battery sizing for storm-resilient photovoltaic
  power island systems,'' \emph{Solar energy}, vol. 109, pp. 165--173, 2014.

\bibitem{olatomiwa2016energy}
L.~Olatomiwa, S.~Mekhilef, M.~S. Ismail, and M.~Moghavvemi, ``Energy management
  strategies in hybrid renewable energy systems: A review,'' \emph{Renewable
  and Sustainable Energy Reviews}, vol.~62, pp. 821--835, 2016.

\bibitem{oh2020theoretical}
E.~Oh and S.-Y. Son, ``Theoretical energy storage system sizing method and
  performance analysis for wind power forecast uncertainty management,''
  \emph{Renewable Energy}, 2020.

\bibitem{black1973pricing}
F.~Black and M.~Scholes, ``The pricing of options and corporate liabilities,''
  \emph{Journal of political economy}, vol.~81, no.~3, pp. 637--654, 1973.

\bibitem{gardiner2009stochastic}
C.~Gardiner, \emph{Stochastic methods}.\hskip 1em plus 0.5em minus 0.4em\relax
  Springer Berlin, 2009, vol.~4.

\end{thebibliography}
\end{document}